\documentclass[runningheads]{llncs}
\usepackage[T1]{fontenc}
\usepackage{graphicx}
\usepackage{comment}
\usepackage{mystyle}
\usepackage{amsmath,amssymb,latexsym,bbm}
\usepackage{algorithm}
\usepackage{algpseudocode}
 \usepackage{booktabs}
\spnewtheorem{myclaim}{Claim}{\itshape}{}

\begin{document}
\title{On the Virtual Network Embedding polytope}
%
\author{Amal Benhamiche \inst{1} \and
Pierre Fouilhoux \inst{2} \and
Lucas Létocart \inst{2} \and Nancy Perrot \inst{1} \and Alexis Schneider \inst{1,2}}
\authorrunning{A. Benhamiche et al.}
%
\institute{Orange Research, F-92320 Châtillon, France \\
\email{\{amal.benhamiche, nancy.perrot, alexis.schneider\}@orange.com} \and
Université Sorbonne Paris Nord, CNRS,
Laboratoire d’Informatique de Paris Nord, LIPN,
F-93430 Villetaneuse, France \\
\email{\{letocart, fouilhoux, schneider\}@lipn.univ-paris13.fr}}
\maketitle              
\begin{abstract}
We initiate the polyhedral study of the Virtual Network Embedding (VNE) problem, which arises in modern telecommunication networks. We propose new valid inequalities for the so-called flow formulation. We then prove, through a dedicated flow decomposition algorithm, that these inequalities characterize the VNE polytope in the case of an embedding of a virtual edge on a substrate path. Preliminary experiments show that the new inequalities propose promising speedups for MIP solvers.

\end{abstract}
%
%
%


\section{Introduction}

Network virtualization is an emerging paradigm that enables Telecommunication Operators to share their infrastructure between multiple virtual networks \cite{anderson2005overcoming,chowdhury2011vineyard}. It is at the heart of network slicing in 5G networks, where each virtual network corresponds to a service and is tailored to meet its requirements, such as speed, security, or coverage \cite{alliance2016description}. Network slicing is essential for deploying new, promising services with specific characteristics, such as remote surgeries or autonomous vehicles.

The core combinatorial structure of network virtualization is known as the Virtual Network Embedding (VNE) problem.
The VNE problem can be defined as follows. 
We consider simple, connected, undirected graphs. 
Let us denote the virtual network as $\Graph_r = (V_r, E_r)$ with $n_r$ nodes; and the substrate network as $\Graph_s = (V_s, E_s)$ with $n_s$ nodes. 
Allocating the substrate resources to the virtual demands is to find a \textit{mapping} (also called an embedding) as follows:

\begin{definition}
    A mapping $m = (m_V, m_E)$ of $\Graph_r$ on $\Graph_s$ is a pair of functions where: 
   \begin{itemize}
       \item $m_V: V_r \rightarrow V_s$ is called the \textit{node placement};
       \item $m_E: E_r \rightarrow P_s$, is called the \textit{edge routing}, where $P_s$ is the set of loop-free paths of $\Graph_s$ and for each $\ebar = (\ubar,\vbar) \in E_r$, $m_E(\ebar)$ is a path of $\Graph_s$ whose endpoints are $m_V(\ubar)$ and $m_V(\vbar)$. 
   \end{itemize}
\end{definition}

In this work, we consider a \textit{one-to-one} node placement: each virtual node has to be placed on a different substrate node. A virtual node $\ubar \in V_r$ (resp. virtual edge $\ebar \in E_r$) has an integer demand $d_{\ubar}$ (resp. $d_{\ebar}$). A substrate node $u \in V_s$ (resp. a substrate edge $e \in E_s$) has an integer capacity  $c_u$ (resp. $c_e$). 
A mapping is said to be \textit{feasible} when the capacity constraints are satisfied, i.e., the sum of the demands of virtual nodes (resp., edges) being placed on a substrate node (resp., being routed on a substrate edge) is less than or equal to the capacity of this node (resp., edge). 

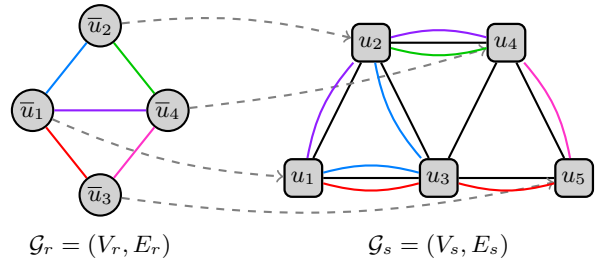
\begin{figure}
    \centering
    \begin{tikzpicture}[scale=0.9]
    \begin{scope}
        \node at (2,-2) {$\Graph_r = (V_r, E_r)$};
        \node[vnode] (v1) at (1,0){$\ubar_1$};
        \node[vnode] (v2) at (2,1.25){$\ubar_2$};
        \node[vnode] (v3) at (2,-1.25){$\ubar_3$};
        \node[vnode] (v4) at (3,0){$\ubar_4$};
    \end{scope}
    \draw[networkedge, blue] (v1) -- (v2);
    \draw[networkedge, purple] (v1) -- (v4);
    \draw[networkedge, green] (v2) -- (v4);
    \draw[networkedge, red] (v1) -- (v3);
    \draw[networkedge, pink] (v3) -- (v4);

    \begin{scope}[xshift=5.cm]
        \node at (2,-2) {$\Graph_s = (V_s, E_s)$};
        \node[snode] (u1) at (0,-1){$u_1$};
        \node[snode] (u2) at (1,1){$u_2$};
        \node[snode] (u3) at (2,-1){$u_3$};
        \node[snode] (u4) at (3,1){$u_4$};
        \node[snode] (u5) at (4,-1){$u_5$};
    \end{scope}
    \draw[networkedge] (u1) -- (u2);
    \draw[networkedge] (u1) -- (u3);
    \draw[networkedge] (u2) -- (u4);
    \draw[networkedge] (u2) -- (u3);
    \draw[networkedge] (u3) -- (u4);
    \draw[networkedge] (u4) -- (u5);
    \draw[networkedge] (u3) -- (u5);

    \draw[placementedge] (v1) to[bend right=10] (u1);
    \draw[placementedge] (v2) to[bend left = 10] (u2);
    \draw[placementedge] (v3) to[bend right= 10] (u5);
    \draw[placementedge] (v4) to[bend right = 5] (u4);

    \draw[routingedge, blue] (u2) to[bend right=15] (u3);
    \draw[routingedge, blue] (u3) to[bend right=15] (u1);
    \draw[routingedge, pink] (u4) to[bend left=15] (u5);
    \draw[routingedge, red] (u1) to[bend right=15] (u3);
    \draw[routingedge, red] (u3) to[bend right=15] (u5);
    \draw[routingedge, purple] (u4) to[bend right=15] (u2);
    \draw[routingedge, purple] (u2) to[bend right=15] (u1);
    \draw[routingedge, green] (u2) to[bend right=15] (u4);
        
    \end{tikzpicture}
    \caption{Example of a mapping of a virtual network (on the left) on a substrate network (on the right).}
    \label{fig:section1:example-vne}
\end{figure}

Figure \ref{fig:section1:example-vne} illustrates the embedding of a virtual graph on the left over a substrate graph on the right. The dotted arrows show the placement of virtual nodes. A virtual edge of a given color is routed using a substrate path of the same color. In this example, $m_V(\ubar_1) = u_1$, $m_V(\ubar_3) = u_5$ and $m_E(\ubar_1, \ubar_3) = \{(u_1, u_3), (u_3, u_5)\}$.

Using a substrate node unit (resp. edge) induces a usage cost $w_u$ (resp. $w_e$). The cost of a mapping $m$, denoted $W_m$, corresponds to the sum of the placement and routing costs: $W_m = \sum_{\ubar \in V_r} d_{\ubar} w_{m_V(\ubar)} + \sum_{\ebar \in E_r} \sum_{e \in m_E(\ebar)} d_{\ebar} w_e$. The VNE problem can be formulated as follows:

\begin{definition}
    Given a virtual network $\Graph_r = (V_r, E_r)$ with demands $d$, a substrate network $\Graph_s = (V_s, E_s)$ with capacities $c$ and costs $w$, the Virtual Network Embedding Problem (VNE) is to find the minimum cost embedding of $\Graph_r$ on $\Graph_s$
\end{definition}

The \vne problem is known to be \np-complete \cite{amaldi2016complexity,rost2020hardness}, even with uniform demands \cite{benhamiche2025complexity}. 
Heuristics and metaheuristics have been heavily investigated for the problem, see the survey in \cite{fischer2013virtual}. In contrast, exact methods such as integer programming have been less considered. A compact flow formulation is proposed for the directed case in \cite{chowdhury2011vineyard,melo2013optimal}, but it is known to suffer from weak linear relaxation, as discussed in \cite{moura2018branch}. In this paper, we exhibit new valid inequalities for this natural formulation and show that they permit to characterize the VNE polytope in the case of a virtual edge and a substrate path.


\section{Flow Formulation and VNE polytope}

The VNE Flow Formulation has been described for directed networks \cite{chowdhury2011vineyard,melo2013optimal}. In the following, we will adapt it for the undirected case by introducing bidirected flow variables for each substrate edge.

\subsection{The undirected Flow Formulation}

Although the edges are undirected, we consider that they are ordered sets (i.e., an arbitrary orientation is given to them), in the aim of formulating VNE as a flow problem.
Thus, the path routing a virtual edge can contain a substrate edge $(u,v) \in E_s$ in either directions $(u,v)$ or $(v,u)$, since the network is undirected.
To capture this, we introduce $\Graph'_s = (V_s, E'_s)$  the bi-directed version of the substrate network, where $E'_s = \bigcup_{(u,v)\in E_s} \{(u,v),(v,u)\}$. 
This bi-orientation technique is standard in flow formulations \cite{ahuja1988flow,raack2011cut}. Additionally, for a substrate node $u \in V_s$, let $\delta^+(u)$ (resp. $\delta^-(u)$) be the outgoing (ingoing) edges of node $u$ in $E'_s$. 

Two sets of binary variables are considered. The placement variable $x_{\ubar u}$ takes the value 1 if and only if the virtual node $\ubar \in V_r$ is placed on substrate node $u \in V_s$. The flow variable $y_{\ebar (u,v)}$ (resp. $y_{\ebar (v,u)}$) takes value 1 if and only if the path routing virtual edge $\ebar \in E_r$ uses substrate edge $(u,v) = e \in E_s$, from $u$ to $v$ (resp. from $v$ to $u$). 
Note that an easy pre-treatment is to set to 0 the variable $x_{\ubar u}$ (resp. $y_{\ebar e}$) if $d_{\ubar} > c_{u}$ (resp. $d_{\ebar e} > c_{e}$), for any $\ubar \in V_r, u \in V_s$ (resp. any $\ebar \in E_r, e \in E_s$). 
In what follows, we thus consider that $d_{\ebar} \le c_e$, for any $e \in E_s$; and $d_{\ubar} \le c_u$, $d_{\vbar} \le c_u$, for any $u \in V_s$.

The Flow Formulation $(FF)$ is the following.

\begin{subequations}
\footnotesize
\begin{align}
    \min \quad & \sum_{u \in V_s} \sum_{\ubar \in V_r}  d_{\ubar} w_u  x_{\ubar u} + \sum_{e \in E'_s} \sum_{\ebar \in E_r} d_{\ebar} w_e y_{\ebar e} \notag \\
    \text{s.t.} \quad 
    & \sum_{u \in V_s} x_{\ubar u} = 1 && \forall \ubar \in V_r  \label{ff1}\\
    & x_{\ubar u} - x_{\vbar u} = \sum_{e \in \delta^+(u)} y_{\ebar e} - \sum_{e \in \delta^-(u)} y_{\ebar e} && \forall (\ubar, \vbar) = \ebar \in E_r,\ \forall u \in V_s \label{ff2} \\
    & \sum_{\ubar \in V_r} x_{\ubar u} \leq 1 && \forall u \in V_s \label{ff3} \\
    & \sum_{\ebar \in E_r} d_{\ebar} (y_{\ebar (u,v)} + y_{\ebar (v,u)}) \leq c_e && \forall (u,v) \in E_s \label{ff5} \\
    & x_{\ubar u} \in \{0,1\} && \forall \ubar \in V_r, \forall u \in V_s,  \nonumber \\
    & y_{\ebar e} \in \{0,1\} && \forall \ebar \in E_r, \forall e \in E'_s \nonumber
\end{align}
\end{subequations}

The constraints (\ref{ff1}) ensure valid virtual node placement. The constraints (\ref{ff2}) are the flow conservation constraints and guarantee that each virtual edge is associated with a routing path in the substrate graph whose end nodes host its endpoints. Constraints (\ref{ff3}) ensure one-to-one node placement.
Constraints (\ref{ff5}) ensure that node edge capacities are respected.
Finally, the objective function expresses an embedding cost taking into account placement and routing costs.

\paragraph{The VNE polytope} For each mapping $m \in \Mapping$, where $\Mapping$ is the set of feasible mappings of $\Graph_r$ on $\Graph_s$, the incidence vector of $m$ (with variables $x$ and $y$) is denoted $\chi(m)$. We are interested in the description of the VNE polytope $\Fpolytope$, i.e., the convex hull of valid mappings of $\Graph_r$ on $\Graph_s$:
\begin{equation}
    \Fpolytope(\Graph_r, \Graph_s, d, c) = conv\{\chi(m), m \in \Mapping\}
\end{equation}

\subsection{Valid inequalities and equalities}

The fact that the Flow Formulation has a poor linear relaxation has been observed in \cite{moura2018branch}. 
In what follows, we propose the first valid inequalities for the VNE problem, to reinforce the formulation.

The following {\it flow departure inequalities} translate that, when a virtual node $\ubar$ is placed on a substrate node $u \in V_s$, then any virtual edge $(\ubar, \vbar) = \ebar \in E_r$ must be routed on an outgoing arc of $u$ in the bidirected version of the substrate network.

\begin{proposition}
    For any $\ebar \in E_r$, for any $u \in V_s$, the flow-departure inequality is valid: 
    \begin{equation} \label{cons-flow-departure}
        x_{\ubar u} \le \sum_{e \in \delta^{+}(u)} y_{\ebar e}
    \end{equation}
\end{proposition}

The following {\it flow continuity inequalities} translate that, when a virtual edge $(\ubar, \vbar) = \ebar \in E_r$ is routed on a substrate edge $(u, v) = e \in E'_s$, then either $\vbar$ is placed on $v$, or the path routing $\ebar$ contains another outgoing arc of $v$, that is not $(v, u)$.

\begin{proposition}
    For any $\ebar = (\ubar, \vbar) \in E_r$, for any $e = (u,v) \in E'_s$, the flow continuity inequality is valid:
    \begin{equation} \label{cons-flow-continuity} 
        y_{\ebar e} \le \sum_{e' \in \delta^+(v)\setminus (v,u)} y_{\ebar e'} + x_{\vbar v} 
    \end{equation}
\end{proposition}

Finally, let $L_s$ be the leaf nodes, i.e., the nodes with only one neighbor, of $\Graph_s$. For a leaf node $l \in L_s$, let $v_l$ its only neighbor in $V_s$. 
Interestingly, on such a leaf node, combining the flow departure and continuity inequalities with the flow conservation constraints results in the following \textit{leaf equality}:

\begin{proposition}
    For any virtual edge $(\ubar, \vbar) = \ebar \in E_r$, for any leaf node $l \in L_s$, the leaf equality is valid:
    \begin{equation} \label{cons-leaf-eq}
        y_{\ebar (l,v_l)} = x_{\ubar l}
    \end{equation}
\end{proposition}

\subsection{Case of a single virtual edge}

In the remainder of the paper, we consider the embedding of a single virtual edge $\ebar = (\ubar, \vbar)$, which is the first natural step to study the VNE polytope. This case can be solved in polynomial time with the modified shortest path algorithm described in \cite{moura2018branch}. 


Interestingly, in that case, the $x$ variables play the role of arc-flow variables: for $u \in V_s$, the variable $x_{\ubar u}$ (resp. $x_{\vbar u}$) corresponds to a flow variable for $\ebar$ on arc $(\ubar, u)$ (resp. $(u, \vbar)$). The node placement inequalities then become flow inequalities constraints on the virtual nodes. 

This means that $(x,y)$ is a $\ubar, \vbar$-arc flow in an \textit{augmented network}. This augmented network is the bidirected version of the substrate network $\Graph'_s$, plus additional arcs: from $\ubar$ to $u$ and from $u$ to $\vbar$, for any $u \in V_s$. We denote this augmented network as $\Graph_s^a = (V_s \cup V_r, E_s^a)$.
Figure~\ref{fig:augmented-network} shows the augmented network (right) of a simple VNE instance (left).

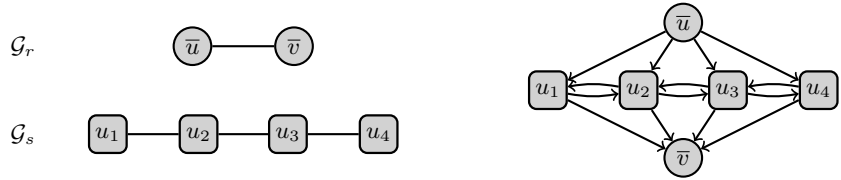
\begin{figure}[h]
\centering
    \begin{tikzpicture}[scale=0.9]

    \begin{scope}[xshift=-3.25cm]
    
        \begin{scope}[yshift=0.65cm]
            \node at (-3.25,0) {$\Graph_r$};
            \node[vnode] (v1) at (-0.75,0){$\ubar$};
            \node[vnode] (v2) at (0.75,0){$\vbar$};
        \end{scope}
        \draw[networkedge] (v1) to (v2);
    
        \begin{scope}[yshift=-0.65cm]
            \node at (-3.25,0) {$\Graph_s$};
            \node[snode] (u1) at (-2,0){$u_1$};
            \node[snode] (u2) at (-0.66,0){$u_2$};
            \node[snode] (u3) at (0.66,0)
            {$u_3$};
            \node[snode] (u4) at (2,0){$u_4$};
        \end{scope}
        \draw[networkedge] (u1) to (u2);
        \draw[networkedge] (u2) to (u3);
        \draw[networkedge] (u3) to (u4);
        
    \end{scope}

    \begin{scope}[xshift=3.25cm]
    \node[vnode] (v1) at (0,1.){$\ubar$};
    \node[vnode] (v2) at (0,-1.){$\vbar$};

    \node[snode] (u1) at (-2.,0){$u_1$};
    \node[snode] (u2) at (-0.66,0){$u_2$};
    \node[snode] (u3) at (0.66, 0){$u_3$};
    \node[snode] (u4) at (2.,0){$u_4$};
        
        \draw[networkedgedir] (v1)  to (u1);
        \draw[networkedgedir] (v1) to (u2);
        \draw[networkedgedir] (v1) to (u3);
        \draw[networkedgedir] (v1) to (u4);
    
        \draw[networkedgedir] (u1) to (v2);
        \draw[networkedgedir] (u2) to (v2);
        \draw[networkedgedir] (u3) to (v2);
        \draw[networkedgedir] (u4) to (v2);
        
        \draw[networkedgedir] (u1) to[bend right=10] (u2);
        \draw[networkedgedir] (u2) to[bend right=10] (u1);
        \draw[networkedgedir] (u2) to[bend right=10] (u3);
        \draw[networkedgedir] (u3) to[bend right=10] (u2);
        \draw[networkedgedir] (u3) to[bend right=10] (u4);
        \draw[networkedgedir] (u4) to[bend right=10] (u3);
        
    \end{scope}
    
\end{tikzpicture}
\caption{Example of the augmented network}
\label{fig:augmented-network}
\end{figure}

Compared to a classical flow problem, the VNE has additional one-to-one node placement capacities, which break the total unimodular structure of the constraints matrix (see example below). We can still however use the Flow Decomposition Theorem \cite{ahuja1988flow}, which shows that any arc-flow can be easily decomposed into a path-flow. Since the fractional variables $(\xfrac,\yfrac)$ correspond to a $\ubar-\vbar$ arc-flow, it translates into the following for VNE with a single virtual edge:

\begin{proposition}
    Any fractional solution $(\tilde{x}, \tilde{y})$ to (FF) can be decomposed into $\ubar-\vbar$ paths $P$ and cycles $C$ of the augmented network $\Graph^a_s$:
    \begin{equation*}
        (\tilde{x}, \tilde{y}) = \sum_{p \in P} \lambda_p \chi(p) + \sum_{c \in C} \lambda_c \chi(c)
    \end{equation*}
    such that $\sum_{p \in P} \lambda_p = 1$. 
\end{proposition}

The \textit{valid paths} of the augmented network $\Graph^a_s$ are paths that use at least three edges of $E^a_s$ (or one in the original substrate network). Thus, a valid path is an integer solution (a feasible mapping). This allows us to devise a tool for characterizations of the VNE polytope when the virtual network is a single edge.

\begin{proposition} \label{prop:integral-property}
    Consider a set of valid inequalities $Ax \le b$.
    If, for any fractional solution of the Flow Formulation reinforced with the inequalities, there exists a $\ubar-\vbar$ path and cycle decomposition on the augmented network, such that all the paths are valid and there are no cycles, then (FF) with these inequalities characterizes the VNE polytope, i.e:
    \begin{equation*}
        \Fpolytope(\Graph_r, \Graph_s, d, c) = \{x, y \in [0, 1], s.t.\ \ref{ff1} - \ref{ff5}, Ax \le b \}
    \end{equation*}
\end{proposition}

\begin{proof}
    In that case, any fractional solution $(\xfrac, \yfrac)$ can be expressed as a convex combination of feasible integer solutions, which implies that all extreme points of $\{x, y \in [0, 1], s.t.\ \ref{ff1} - \ref{ff5}, Ax \ge b \}$ are integral. \qed
\end{proof}

Figure~\ref{fig:ex-bad-relax2} illustrates an example of a flow decomposition (right) of a fractional point (left) for a simple edge on path instance. The fractional point $(\xfrac, \yfrac)$ is the following: $\xfrac_{\ubar u_1} = \xfrac_{\ubar u_4} = \xfrac_{\vbar u_1} = \xfrac_{\vbar u_4} = 0.5,\ \xfrac_{\ubar u_2} = \xfrac_{\ubar u_3} = \xfrac_{\vbar u_2} = \xfrac_{\vbar u_3} = 0$,\\
$\yfrac_{\ebar (u_1, u_2)} = \yfrac_{\ebar (u_2, u_1)} = \yfrac_{\ebar (u_3, u_4)} = \yfrac_{\ebar (u_4, u_3)}  = 0.5$,\ $\yfrac_{\ebar (u_2, u_3)} = \yfrac_{\ebar (u_3, u_2)} = 0$.

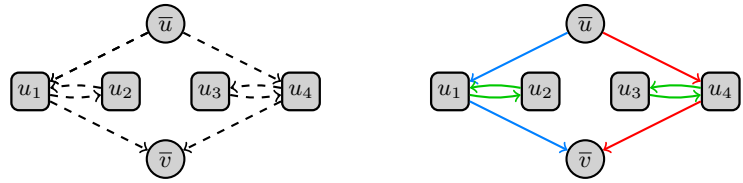
\begin{figure}[h]
\centering
\begin{tikzpicture}[scale=0.9]

    \begin{scope}
    
    \begin{scope}[xshift=-3.1cm]
    \node[vnode] (v1) at (0,1.){$\ubar$};
    \node[vnode] (v2) at (0,-1.){$\vbar$};

    \node[snode] (u1) at (-2.,0){$u_1$};
    \node[snode] (u2) at (-0.66,0){$u_2$};
    \node[snode] (u3) at (0.66, 0){$u_3$};
    \node[snode] (u4) at (2.,0){$u_4$};

    \draw[networkedgedir, black, dashed] (v1) to (u1);
    \draw[networkedgedir, black, dashed] (v1) to (u4);
    \draw[networkedgedir, black, dashed] (u1) to (v2);
    \draw[networkedgedir, black, dashed] (u4) to (v2);
        \draw[networkedgedir, black, dashed] (v1) to (u1);

    \draw[networkedgedir, black, dashed] (u1) to[bend right = 10] (u2);
    \draw[networkedgedir, black, dashed] (u2) to[bend right = 10] (u1);

    \draw[networkedgedir, black, dashed] (u4) to[bend right = 10] (u3);
    \draw[networkedgedir, black, dashed] (u3) to[bend right = 10] (u4);
    
    \end{scope}

    \begin{scope}[xshift=3.1cm]

    \node[vnode] (v1) at (0,1.){$\ubar$};
    \node[vnode] (v2) at (0,-1.){$\vbar$};

    \node[snode] (u1) at (-2.,0){$u_1$};
    \node[snode] (u2) at (-0.66,0){$u_2$};
    \node[snode] (u3) at (0.66, 0){$u_3$};
    \node[snode] (u4) at (2.,0){$u_4$};

    \draw[networkedgedir, blue] (v1) to (u1);
    \draw[networkedgedir, blue] (u1) to (v2);

    \draw[networkedgedir, red] (v1) to (u4);
    \draw[networkedgedir, red] (u4) to (v2);

    \draw[networkedgedir, green] (u1) to[bend right = 10] (u2);
    \draw[networkedgedir, green] (u2) to[bend right = 10] (u1);

    \draw[networkedgedir, green] (u4) to[bend right = 10] (u3);
    \draw[networkedgedir, green] (u3) to[bend right = 10] (u4);
    
        
    \end{scope}
    \end{scope}

\end{tikzpicture}
\caption{Example of a flow decomposition}
\label{fig:ex-bad-relax2}
\end{figure}

The unique path and cycle decomposition is the following. It contains the two paths: $p_1 = (\ubar, u_1, \vbar)$ (blue on the figure) and $p_2 = (\ubar, u_4, \vbar)$ (red on the figure), with $\lambda_{p_1} = \lambda_{p_2} = 0.5$; as well as two 2-node cycles (green on the figure).

This fractional point respects the flow departure inequalities. Since it contains non-valid paths, this example shows that (FF) with the flow departure inequalities is not integral for the case of a virtual edge on a substrate path. Note that, here, adding the flow continuity inequalities removes any such fractional solution. In the next section, we show that this is indeed always the case.

\section{Characterization of the polytope for a substrate path}

We now show that the Flow Formulation, with the flow departure and the flow continuity inequalities, characterizes the VNE polytope when the substrate network is a path $\Path_s$. 

Consider a fractional solution $(\tilde{x},\tilde{y})$.
We have to construct a flow decomposition of  $(\tilde{x},\tilde{y})$ with only valid paths to prove the result. 
We found that the usual flow decomposition algorithm (described in \cite{ahuja1988flow}) does not permit this. Thus, we provide our own flow decomposition algorithm, called ComputeFlowDecomposition. Given the fractional solution, this Algorihm~\ref{algo:path-decompo} returns a path-cycle decomposition $(P, C, \lambda)$. 

\begin{algorithm}[t!]
\caption{ComputeFlowDecomposition}\label{algo:path-decompo}
\begin{algorithmic}[1]
\Require A fractional solution $(\xfrac, \yfrac)$
\Ensure A flow decomposition $(P, C, \lambda)$ where all paths of $P$ are valid and $C = \emptyset$

\State Initialize empty arrays $P, C, \lambda$
\State Initialize arrays $(\xfrac', \yfrac') \gets (\xfrac, \yfrac)$
\For{$u_i \in (u_1,\ldots, u_{n_s-1} )$}  \label{algo:stage1} \Comment{Forward-mappings}
    \State $j \gets i + 1$
    \While{ $\yfrac'_{\ebar (u_i, u_{i+1})} > 0$ }
        \State $\delta \gets \min(\ \yfrac'_{\ebar (u_i, u_{i+1})},\ \xfrac'_{\vbar u_j} - \max(\yfrac'_{\ebar (u_{j+1}, u_j)} -\ \yfrac'_{\ebar (u_j, u_{j-1})},\ 0))$ \label{eq:algo}
        \If{$\delta > 0$}
            \State $p \gets (\ubar, u_i, \ldots, u_j, \vbar)$
            \State $P \gets P \cup p$
            \State $\lambda_p \gets \delta$
            \State $(\tilde{x}', \tilde{y}') \gets (\tilde{x}', \tilde{y}') -  \delta \cdot \chi(p)$
        \EndIf
        \State $j \gets j + 1$
    \EndWhile
\EndFor
\For{$u_i \in (u_{n_s},\ldots, u_2)$} \Comment{Backward-mappings} \label{algo:stage2}
    \State $j \gets i - 1$
    \While{ $\xfrac'_{\ubar u_i} > 0$ }
        \State $\delta \gets \xfrac'_{\vbar u_j} $
        \If{$\delta > 0$}
            \State $p \gets (\ubar, u_i, \ldots, u_j, \vbar)$
            \State $P \gets P \cup p$
            \State $\lambda_p \gets \delta$
            \State $(\tilde{x}', \tilde{y}') \gets (\tilde{x}', \tilde{y}') - \delta \cdot \chi(p)$
        \EndIf
        \State $j \gets j - 1$
    \EndWhile
\EndFor

\State \Return $(P, C, \lambda)$
\end{algorithmic}
\end{algorithm}

The algorithm incrementally builds a set of valid paths by first building \textit{forward-oriented} paths (paths that start at a node $u_i \in V_s$ and finish at $u_j$ such that $i < j$) in the For Loop of Line~\ref{algo:stage1}; and then \textit{backwards-oriented} paths in the For Loop of Line~\ref{algo:stage2}. The algorithm relies on a \textit{residual fractional solution} $(\xfrac', \yfrac')$, which, at any step of procedure, corresponds to $(\xfrac, \yfrac) - \sum_{p \in P} \lambda_p \chi(p)$. 

An important property is that the residual fractional solution always satisfies the flow departure and continuity inequalities, which guarantees that the algorithm can always progress. This property is maintained during the forward-oriented path phase as follows.
At a node $u_l \in V_s$, enough residual variable $\xfrac'_{\vbar u_l}$ must be preserved to allow the construction of backward-oriented mappings later. Specifically, the required amount is $\yfrac'_{\ebar (u_{l+1}, u_l)} - \yfrac'_{\ebar (u_l, u_{l-1})}$ (if greater than 0): this is enforced by Line~\ref{eq:algo}. \\




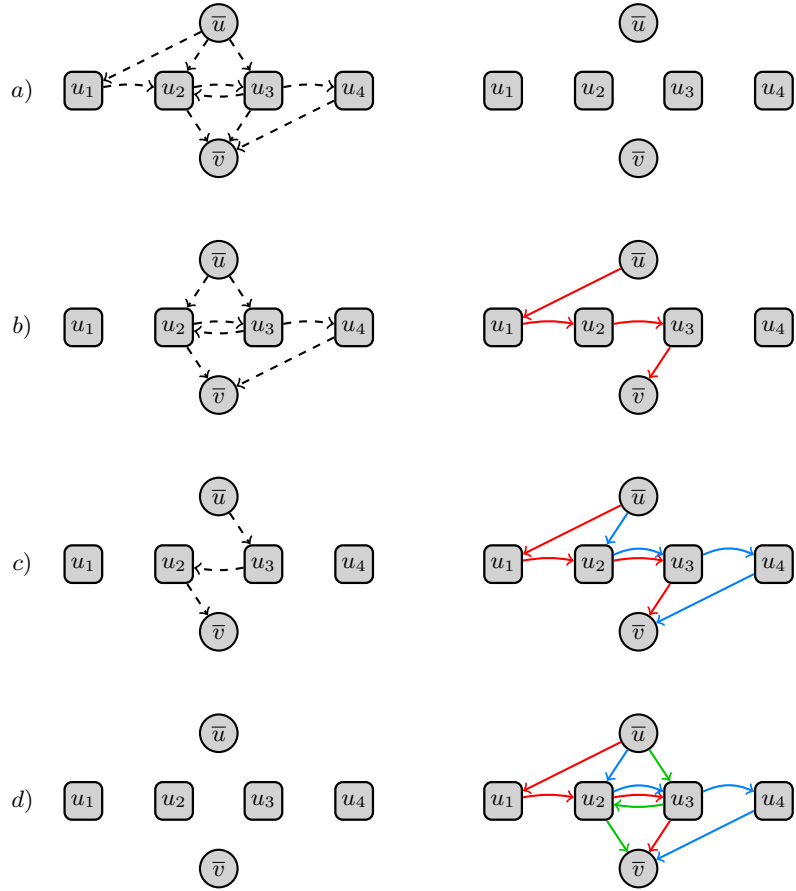
\begin{figure}[t!]
\centering
\begin{tikzpicture}[scale=0.9]

    \begin{scope}[yshift=5cm]
    \node at (-6, 0){$a)$};
    
    \begin{scope}[xshift=-3.1cm]
    \node[vnode] (v1) at (0,1.){$\ubar$};
    \node[vnode] (v2) at (0,-1.){$\vbar$};

    \node[snode] (u1) at (-2.,0){$u_1$};
    \node[snode] (u2) at (-0.66,0){$u_2$};
    \node[snode] (u3) at (0.66, 0){$u_3$};
    \node[snode] (u4) at (2.,0){$u_4$};

    \draw[networkedgedir, black, dashed] (v1) to (u1);
    \draw[networkedgedir, black, dashed] (v1) to (u2);
    \draw[networkedgedir, black, dashed] (v1) to (u3);
    \draw[networkedgedir, black, dashed] (u1) to[bend left = 10] (u2);
    \draw[networkedgedir, black, dashed] (u2) to[bend left = 10] (u3);
    \draw[networkedgedir, black, dashed] (u3) to[bend left = 10] (u4);
    \draw[networkedgedir, black, dashed] (u3) to[bend left = 10] (u2);
    \draw[networkedgedir, black, dashed] (u2) to (v2);
    \draw[networkedgedir, black, dashed] (u3) to (v2);
    \draw[networkedgedir, black, dashed] (u4) to (v2);
    
    \end{scope}

    \begin{scope}[xshift=3.1cm]

    \node[vnode] (v1) at (0,1.){$\ubar$};
    \node[vnode] (v2) at (0,-1.){$\vbar$};

    \node[snode] (u1) at (-2.,0){$u_1$};
    \node[snode] (u2) at (-0.66,0){$u_2$};
    \node[snode] (u3) at (0.66, 0){$u_3$};
    \node[snode] (u4) at (2.,0){$u_4$};

        
    \end{scope}
    \end{scope}

    \begin{scope}[yshift=1.5cm]
    \node at (-6, 0){$b)$};
    
    \begin{scope}[xshift=-3.1cm]
    \node[vnode] (v1) at (0,1.){$\ubar$};
    \node[vnode] (v2) at (0,-1.){$\vbar$};

    \node[snode] (u1) at (-2.,0){$u_1$};
    \node[snode] (u2) at (-0.66,0){$u_2$};
    \node[snode] (u3) at (0.66, 0){$u_3$};
    \node[snode] (u4) at (2.,0){$u_4$};

    \draw[networkedgedir, black, dashed] (v1) to (u2);
    \draw[networkedgedir, black, dashed] (v1) to (u3);
    \draw[networkedgedir, black, dashed] (u2) to[bend left = 10] (u3);
    \draw[networkedgedir, black, dashed] (u3) to[bend left = 10] (u4);
    \draw[networkedgedir, black, dashed] (u3) to[bend left = 10] (u2);
    \draw[networkedgedir, black, dashed] (u2) to (v2);
    \draw[networkedgedir, black, dashed] (u4) to (v2);
    
    \end{scope}

    \begin{scope}[xshift=3.1cm]

    \node[vnode] (v1) at (0,1.){$\ubar$};
    \node[vnode] (v2) at (0,-1.){$\vbar$};

    \node[snode] (u1) at (-2.,0){$u_1$};
    \node[snode] (u2) at (-0.66,0){$u_2$};
    \node[snode] (u3) at (0.66, 0){$u_3$};
    \node[snode] (u4) at (2.,0){$u_4$};

    \draw[networkedgedir, red] (v1) to (u1);
    \draw[networkedgedir, red] (u1) to[bend left = 10] (u2);
    \draw[networkedgedir, red] (u2) to[bend left = 10] (u3);
    \draw[networkedgedir, red] (u3) to[bend left = 0] (v2);

    \end{scope}
    \end{scope}

    \begin{scope}[yshift=-2.cm]
    \node at (-6, 0){$c)$};
    
    \begin{scope}[xshift=-3.1cm]
    \node[vnode] (v1) at (0,1.){$\ubar$};
    \node[vnode] (v2) at (0,-1.){$\vbar$};

    \node[snode] (u1) at (-2.,0){$u_1$};
    \node[snode] (u2) at (-0.66,0){$u_2$};
    \node[snode] (u3) at (0.66, 0){$u_3$};
    \node[snode] (u4) at (2.,0){$u_4$};

    \draw[networkedgedir, black, dashed] (v1) to (u3);
    \draw[networkedgedir, black, dashed] (u3) to[bend left = 10] (u2);
    \draw[networkedgedir, black, dashed] (u2) to (v2);
    
    \end{scope}

    \begin{scope}[xshift=3.1cm]

    \node[vnode] (v1) at (0,1.){$\ubar$};
    \node[vnode] (v2) at (0,-1.){$\vbar$};

    \node[snode] (u1) at (-2.,0){$u_1$};
    \node[snode] (u2) at (-0.66,0){$u_2$};
    \node[snode] (u3) at (0.66, 0){$u_3$};
    \node[snode] (u4) at (2.,0){$u_4$};

    \draw[networkedgedir, red] (v1) to (u1);
    \draw[networkedgedir, red] (u1) to[bend left = 10] (u2);
    \draw[networkedgedir, red] (u2) to[bend left = 10] (u3);
    \draw[networkedgedir, red] (u3) to[bend left = 0] (v2);

    \draw[networkedgedir, blue] (v1) to (u2);
    \draw[networkedgedir, blue] (u2) to[bend left = 25] (u3);
    \draw[networkedgedir, blue] (u3) to[bend left = 25] (u4);
    \draw[networkedgedir, blue] (u4) to[bend left = 0] (v2);

    \end{scope}
    \end{scope}

    \begin{scope}[yshift=-5.5cm]
    \node at (-6, 0){$d)$};
    
    \begin{scope}[xshift=-3.1cm]
    \node[vnode] (v1) at (0,1.){$\ubar$};
    \node[vnode] (v2) at (0,-1.){$\vbar$};

    \node[snode] (u1) at (-2.,0){$u_1$};
    \node[snode] (u2) at (-0.66,0){$u_2$};
    \node[snode] (u3) at (0.66, 0){$u_3$};
    \node[snode] (u4) at (2.,0){$u_4$};

    
    \end{scope}

    \begin{scope}[xshift=3.1cm]

    \node[vnode] (v1) at (0,1.){$\ubar$};
    \node[vnode] (v2) at (0,-1.){$\vbar$};

    \node[snode] (u1) at (-2.,0){$u_1$};
    \node[snode] (u2) at (-0.66,0){$u_2$};
    \node[snode] (u3) at (0.66, 0){$u_3$};
    \node[snode] (u4) at (2.,0){$u_4$};

    \draw[networkedgedir, red] (v1) to (u1);
    \draw[networkedgedir, red] (u1) to[bend left = 10] (u2);
    \draw[networkedgedir, red] (u2) to[bend left = 10] (u3);
    \draw[networkedgedir, red] (u3) to[bend left = 0] (v2);

    \draw[networkedgedir, blue] (v1) to (u2);
    \draw[networkedgedir, blue] (u2) to[bend left = 25] (u3);
    \draw[networkedgedir, blue] (u3) to[bend left = 25] (u4);
    \draw[networkedgedir, blue] (u4) to[bend left = 0] (v2);

    \draw[networkedgedir, green] (v1) to (u3);
    \draw[networkedgedir, green] (u3) to[bend left = 10] (u2);
    \draw[networkedgedir, green] (u2) to (v2);

    \end{scope}
    \end{scope}

\end{tikzpicture}
\caption{Steps of the algorithm on a 4-node substrate path. On the left, the residual fractional variables; on the right, the valid paths $P$ constructed by the algorithm, a) before the algorithm, b) after the first stage on $u_1$, c) after the first stage on $u_2$, d) after the second stage on $u_3$ (algorithm is completed).}
\label{fig:algo-path}
\end{figure}

We now give an example. Figure~\ref{fig:algo-path} shows the successive steps of Algorithm~\ref{algo:path-decompo} for a 4-node path and a fractional solution $(\xfrac, \yfrac)$, which is the following:\\
\noindent  $\xfrac_{\ubar u_1} = 0.4,\ \xfrac_{\ubar u_2} = 0.3,\ \xfrac_{\ubar u_3} = 0.3,\ \xfrac_{\ubar u_4} = 0.$, \\
\noindent $\xfrac_{\vbar u_1} = 0.,\ \xfrac_{\vbar u_2} = 0.3,\ \xfrac_{\vbar u_3} = 0.4,\ \xfrac_{\vbar u_4} = 0.3$, \\
\noindent $\yfrac_{\ebar (u_1, u_2)} = 0.4,\ \yfrac_{\ebar (u_2, u_1)} = 0.,\ \yfrac_{\ebar (u_2, u_3)} = 0.7,\ \yfrac_{\ebar (u_3, u_2)} = 0.3$, \\
\noindent $\yfrac_{\ebar (u_3, u_4)} = 0.3,\ \yfrac_{\ebar (u_4, u_3)} = 0.$. 

At each step, the non-zero residual fractional variables are shown on the left of the figure, and the valid paths built by the algorithm on the right. At the end, three paths are built: $p_1 = (\ubar, u_1, u_2, u_3, \vbar)$ (in red in the figure), $p_2 = (\ubar, u_2, u_3, u_4, \vbar)$ (in blue), and $p_3 = (\ubar, u_3, u_2, \vbar)$ (in green), while all residual fractional variables have null values.
Note that since $\yfrac_{\ebar (u_3, u_2)} = 0.3$ and $\yfrac_{\ebar (u_2, u_1)} = 0.$, the path $(\ubar, u_1, u_2, \vbar)$ is not built. 
The paths, with $\lambda_{p_1} = 0.4$, $\lambda_{p_2} = 0.3$, and $\lambda_{p_3} = 0.3$, yield a path decomposition of the fractional point such that all paths are valid, with no cycles.  \\

In what follows, we prove that the algorithm works. First, note that the residual variables $\xfrac'$ are always positive (direct), and that, at any point of the algorithm, 
\begin{equation} \label{eq:suwu}
    \sum_{u \in V_s} \xfrac'_{\ubar u} = \sum_{u \in V_s} \xfrac'_{\vbar u} = 1 - \sum_{p \in P} \lambda_p
\end{equation}
We now look at the first path constructed and observe a few useful properties. We then consider later steps of the algorithm, for which the observations easily extend. This permits us to show the algorithm's validity and termination. \\

\noindent \textbf{First path constructed:}
To begin, we will consider the iterations of the algorithm on the first node $u_1$. We suppose that $\xfrac_{\ubar u_1} > 0$ (in the other case, our arguments can be easily adapted) and consider the first path built by the algorithm, say $p = (\ubar, u_1, \ldots, u_j, \vbar)$, with $u_j \in \{u_2, \ldots, u_{n_s -1} \}$. The flow continuity inequalities ensure that such a path exists. We show a series of useful claims for that first iteration.

\begin{myclaim}
    \begin{equation}\label{eq:horrible}
        \yfrac_{\ebar (u_{1}, u_{2})} \le \yfrac_{\ebar (u_l, u_{l+1})},\ \text{for } l \in \{2, \ldots, j-1\}
    \end{equation}
\end{myclaim}

\begin{proof}
We show that $\yfrac_{\ebar (u_{1}, u_{2})} \le \yfrac_{\ebar (u_2, u_{3})}$ (in the case where no path was built finishing on $u_2$). Since no path was built on the node, based on Line~\ref{eq:algo} of the algorithm, two cases are possible:
\begin{itemize}
    \item $\xfrac_{\vbar u_2} = 0$: the result follows from the flow continuity inequality on $(u_{1}, u_2)$.
    \item $\xfrac_{\vbar u_2} \le \yfrac_{\ebar (u_{3} u_2)} - \yfrac_{\ebar (u_{2} u_{1})}$: injecting this in the flow conservation equality on node $u_2$, we obtain $\yfrac_{\ebar (u_{2}, u_{3})} \ge \yfrac_{\ebar (u_{1}, u_{2})} + \xfrac_{\ubar u_2} \ge \yfrac_{\ebar (u_{1}, u_{2})}$.
\end{itemize}
On the following substrate nodes, the proof is analogous. \qed
\end{proof}

\begin{myclaim}
    \begin{equation}
        \yfrac_{\ebar (u_{1}, u_{2})} - \lambda_p \le \yfrac_{\ebar (u_j, u_{j+1})}
    \end{equation}
\end{myclaim}   

\begin{proof}
Based on Line~\ref{eq:algo}, there are three possible cases:
\begin{itemize}
    \item $\lambda_p = \yfrac_{\ebar(u_{1} u_2)}$: direct.
    \item $\lambda_p = \xfrac_{\vbar u_j}$: the result follows from the flow continuity inequality on $(u_{j-1}, u_j)$, and from Equation~(\ref{eq:horrible}).
    \item $\lambda_p = \xfrac_{\vbar u_j} - (\yfrac_{\ebar (u_{j+1} u_j)} - \yfrac_{\ebar (u_{j} u_{j-1})})$: injecting in the flow conservation equality on node $u_j$, we obtain $
        \yfrac_{\ebar (u_{j}, u_{j+1})} = \yfrac_{\ebar (u_{j-1}, u_{j})} - \lambda_p + \xfrac'_{\ubar u_j}$. 
    The result is obtained using Equation~(\ref{eq:horrible}). \qed
\end{itemize}
\end{proof}

The two previous claims mean that the residual $\yfrac'$ variables are still greater than 0 after the substitution of $\lambda_p \chi(p)$. Next, we prove that the additional inequalities are also valid for the residual variables.

\begin{myclaim}
    The flow continuity inequalities remain valid for $(\xfrac, \yfrac) - \lambda_p \chi(p)$.
\end{myclaim}
\begin{proof}
We show this for each inequality:
\begin{itemize}
    \item On $(u_l, u_{l+1})$, for $u_l \in \{u_1, \ldots, u_{j-1} \}$: $\lambda_p$ is removed from both the left-hand side (LHS) and the right-hand side (RHS) of the inequality. 
    \item On $(u_{j-1}, u_j)$: $\lambda_p$ is removed only from the LHS.
    \item On $(u_{j+1}, u_{j})$: due to Line~\ref{eq:algo}, 
        \begin{align*}
            \lambda_p \leq \xfrac_{\vbar u_j} - (\yfrac_{\ebar (u_{j+1}, u_{j})} - \yfrac_{\ebar (u_{j}, u_{j-1})}) \\
            \Leftrightarrow\ \yfrac_{\ebar (u_{j+1}, u_{j})} \leq \xfrac_{\vbar u_j} - \lambda_p + \yfrac_{\ebar (u_{j}, u_{j-1})}
        \end{align*}
    \item On other edges, the inequality is not affected. \qed
\end{itemize}
\end{proof}

\begin{myclaim}
    The flow departure inequalities remain valid for $(\xfrac, \yfrac) - \lambda_p \chi(p)$.
\end{myclaim}

\begin{proof}
    We show this for each inequality:
    \begin{itemize}
        \item On $u_1$: $\lambda_p$ is removed from both LHS and RHS of the inequality. 
        \item On $u_l \in u_{2}, \ldots, u_{j-1}$: since no path was built finishing at $u_l$, then, from Line~\ref{eq:algo} of the algorithm, either:
        \begin{itemize}
            \item $\xfrac_{\vbar u_l} = 0$: the flow conservation equality on $u_l$ becomes
            \[
                \xfrac_{\ubar u_l} + \yfrac_{\ebar (u_{l-1} u_l)} + \yfrac_{\ebar (u_{l+1} u_{l})}  = \yfrac_{\ebar (u_{l} u_{l+1})} + \yfrac_{\ebar (u_{l} u_{l-1})}
            \]
            By Equation~(\ref{eq:horrible}), $\lambda_p \leq \yfrac_{\ebar (u_1, u_{2})} \leq \yfrac_{\ebar (u_{l-1}, u_{l})}$. We obtain:
            \[
                \xfrac_{\ubar u_l} \le \xfrac_{\ubar u_l} +  \yfrac_{\ebar (u_{l+1} u_{l})}  \le \yfrac_{\ebar (u_{l} u_{l+1})} + \yfrac_{\ebar (u_{l} u_{l-1})} - \lambda_p
            \]
            \item $\xfrac_{\vbar u_l} \le \yfrac_{\ebar (u_{l+1} u_{l})} - \yfrac_{\ebar (u_l, u_{l-1})}$: injecting in the flow conservation constraint on $u_l$,
            \[
                \xfrac_{\ubar u_l} + \yfrac_{(u_{l-1} u_l)} \le \yfrac_{(u_{l} u_{l+1})}
            \] 
            Again, using $\lambda_p \leq \yfrac_{\ebar (u_{l-1}, u_{l})}$, we obtain the result.
        \end{itemize}
        \item On other nodes, the inequality is not affected. \qed
    \end{itemize}
\end{proof}

\noindent \textbf{Later iterations on $u_1$:} 
After that, while $\yfrac'_{\ebar (u_1, u_2)} > 0$, the algorithm continues to construct paths starting with $(\ubar, u_1)$. All the previous claims easily extend to these new paths. 
Thus, such a path can always be constructed, since the flow continuity inequalities remain valid for the residual fractional variables.
The residual $\yfrac'$ variables will always remain positive. Eventually, the iterations on nodes $u_1$ are completed, at the latest at node $u_{n_s}$, since $\yfrac_{\ebar (u_1, u_2)} \le \yfrac_{\ebar (u_{n_s-1}, u_{n_s})} = \xfrac_{\vbar u_{n_s}}$ (by Equation~\ref{eq:horrible}). After that, $\yfrac'_{\ebar (u_1, u_2)} = \xfrac'_{\ubar u_1} = 0$. \\ 

\noindent \textbf{Later nodes of forward-oriented paths phase:} Consider, if it exists, the next node $u_i$ of $u_2, \ldots, u_{n_s}$ such that $\yfrac'_{\ebar (u_i, u_{i+1})} > 0$. Since the flow continuity inequalities remain valid for the residual fractional variables at that point, all the previous claims and observations extend to iterations on that node $u_i$. Thus the algorithm constructs forward paths starting with $(\ubar, u_i)$ until $\yfrac'_{\ebar (u_2, u_3)} = 0$. Note that $\xfrac'_{\ubar u_i}$ may then be non-null. The algorithm continues so on later nodes, until the end of iterations of forward-oriented paths phase, at which point $\yfrac'_{\ebar (u_l, u_{l+1})} = 0$, for all $u_l \in \{u_1, \ldots, u_{n_s-1}\}$.\\  

\noindent \textbf{Backward-oriented paths phase:} The algorithm then iterates over $u_{n_s}, \ldots, u_2$ and construct backward-oriented paths. Consider the first node $u_i$ of those such that, at that point, $\xfrac'_{\ubar u_i} > 0$. Since at that point, the residual fractional solution respects the flow departure and continuity inequalities (Claim~3-4), paths can also be constructed until the variables $\xfrac'_{\ubar u_i}$ and $\yfrac'_{\ebar (u_i, u_{i-1})}$ have null values. \\

At the end of the algorithm, all residual variables will have a null value. The set of valid paths $P$ constructed is then a flow decomposition for solution $(\xfrac, \yfrac)$ such that all paths are valid and there are no cycles, with $\sum_{p \in P} \lambda_p  = 1$ (by Equation~(\ref{eq:suwu})). Using Proposition~\ref{prop:integral-property}, the integrality is proven:

\begin{theorem} \label{thm:integrality-path}
    For a virtual edge $\ebar$ on a substrate path $\Path_s$, the Flow Formulation with the flow departure and continuity inequalities characterizes the VNE polytope:
    \begin{equation*}
        \Fpolytope(\ebar, \Path_s, d, c) = \{x, y \in [0, 1],\ s.t.\ (\ref{ff1}) - (\ref{ff5}),\ (\ref{cons-flow-departure}),\ (\ref{cons-flow-continuity}) \}
    \end{equation*}
\end{theorem}

\section{Computational results}

We now run preliminary experiments to study the effectiveness of our new inequalities. We use the solver IBM ILOG CPLEX 22.1, on a machine with an Intel\textsuperscript{\textregistered} Xeon\textsuperscript{\textregistered} E5-2650 v3 CPU running at 2.30 GHz. 

The virtual networks are 100 (connected) Erdös–Rényi networks \cite{erdds1959random}, with 14 nodes and 22 edges, with unit demands on every node and edge.
The substrate networks are Intellifiber, Uninett, and TW (70-75 nodes), which are real network topologies from the Internet Topology Zoo \cite{knight2011internet}. The capacities are generated randomly for edges, and for nodes, we consider scenarios where 25\%, 50\%, and 100\% of nodes have a capacity of 1 (the rest 0), to test the efficiency of inequalities in more sparse capacity scenarios. The costs are generated randomly.

\vspace{-0.1cm}

\begin{table*}[h!]
\centering
\scriptsize
\begin{tabular*}{\textwidth}{@{\extracolsep\fill}llrrr@{}}\toprule
    \textbf{Capacities} & \textbf{Constraints} & \textbf{Time} & \textbf{Nodes} & \textbf{v(LP)} \\ \midrule
  0.25 & FF & 46.5  & 11564.6 & 29.6 \\
   & FF + FD & 22.4 & 7440.5 & 72.1 \\
   & FF + FD + FC & 24.6  & 4740.8 &  86.4 \vspace{0.05cm} \\ 
  0.5 & FF & 107.0  & 25638.5 &  20.4 \\
  & FF + FD & 38.4  & 12683.6 &  51.8 \\
   & FF + FD + FC & 65.1  & 9124.3 &  52.7 \vspace{0.05cm} \\
  1.0 & FF & 211.1 & 38124.9 &  14.7 \\
   & FF + FD & 106.6 &  16377.0 &  40.9 \\
   & FF + FD + FC & 224.4 &  12340.2 &  40.9 \\
    \bottomrule
  \end{tabular*}
\caption{Experimental results} \label{tab:abilene}
\end{table*}

\vspace{-0.4cm}

The results are reported in Table \ref{tab:abilene}. Each row gives the average CPLEX performance given a set of inequalities for a given percentage of capacitated nodes (first column).
The remaining columns report solving time (seconds), number of Branch\&Bound nodes, and value of the linear relaxation.

The flow departure (FD) inequalities are very effective, providing 2x to 3x speedups and largerelaxation gains. Note that, when using small real virtual networks such as Abilene and Polska from \cite{knight2011internet}, we obtained even stronger gains (up to 20x). Flow continuity (FC) inequalities also improve B\&B tree sizes and linear relaxations (in the sparse capacity scenario), but overall, they slow down the MIP solver. This is because, although polynomial, the number of these inequalities is large (more than 4000). Adding only a small subset of these inequalities, e.g., through a Branch\&Cut approach, could make them practically effective.

\section{Conclusion}

In this paper, we proposed a first study of the Virtual Network Embedding problem polytope. We introduced new valid inequalities and proved that they characterize the polytope for the case of a virtual edge on a substrate path. Experiments showed that inequalities help MIP solvers. We believe that the theoretical result could be extended to substrate trees by adapting our proof. Then new valid inequalities associated to larger virtual network subgraphs than an edge (e.g., stars) should be studied to further improve MIP performance.

%
%
\bibliographystyle{splncs04}
\bibliography{biblio}

\end{document}